\documentclass[runningheads]{llncs}

\usepackage[T1]{fontenc}
\usepackage{makeidx} 
\usepackage{amsmath,mathtools}
\usepackage{amssymb}
\usepackage{array}
\usepackage[bookmarks=false]{hyperref}
\usepackage{float}
\usepackage{graphicx}
\usepackage[export]{adjustbox}
\usepackage{macros}
\usepackage{subcaption}
\usepackage{multirow}
\usepackage{tikz}
\usepackage{color}
\usepackage{xcolor} 
\usepackage{colortbl}
\usepackage{cancel}
\usepackage{relsize}
\usepackage{adjustbox}
\usepackage{orcidlink}
\usepackage[vlined,boxed,linesnumbered]{algorithm2e}
\usepackage{listings}
\usepackage{shuffle}
\usepackage{enumitem}

\definecolor{darkgreen}{rgb}{0.0, 0.5, 0.0}

\expandafter\def\expandafter\normalsize\expandafter{%
    \normalsize%
    \setlength\abovedisplayskip{3pt}%
    \setlength\belowdisplayskip{3pt}%
    \setlength\abovedisplayshortskip{3pt}%
    \setlength\belowdisplayshortskip{3pt}%
}

\setlist[itemize]{topsep=3pt}  
\setlist[enumerate]{topsep=3pt}

\usepackage{lineno}

\usepackage{graphicx}
\usepackage{color}

\begin{document}

\title{Translating Workflow Nets into the Partially Ordered Workflow Language}
\titlerunning{Translating Workflow Nets into POWL}

\author{Humam Kourani\inst{1,2}\orcidID{0000-0003-2375-2152} \and
Gyunam Park\inst{1}\orcidID{0000-0001-9394-6513} \and
Wil M.P. van der Aalst\inst{1,2}\orcidID{0000-0002-0955-6940}}
\authorrunning{H. Kourani et al.}
%
\institute{Fraunhofer Institute for Applied Information Technology FIT, Schloss Birlinghoven, 53757 Sankt Augustin, Germany\\
\email{\{humam.kourani,gyunam.park,wil.van.der.aalst\}@fit.fraunhofer.de} \and
RWTH Aachen University, Ahornstraße 55, 52074 Aachen, Germany}
\maketitle              
\begin{abstract}
The Partially Ordered Workflow Language (POWL) has recently emerged as a process modeling notation, offering strong quality guarantees and high expressiveness. However, its adoption is hindered by the prevalence of standard notations like workflow nets (WF-nets) and BPMN in practice. This paper presents a novel algorithm for transforming safe and sound WF-net into equivalent POWL models. The algorithm recursively identifies structural patterns within the WF-net and translates them into their POWL representation. We formally prove the correctness of our approach, showing that the generated POWL model preserves the language of the input WF-net. Furthermore, we demonstrate the high scalability of our algorithm, and we show its completeness on a subclass of WF-nets that encompasses equivalent representations for all POWL models. This work bridges the gap between the theoretical advantages of POWL and the practical need for compatibility with established notations, paving the way for broader adoption of POWL in process analysis and improvement applications.

\end{abstract}

\section{Introduction}\label{sec:intro}
The field of process modeling and analysis relies heavily on formal notations to represent and reason about the behavior of complex systems. While standard notations like Petri nets \cite{DBLP:journals/topnoc/HeeSW13a}, and specifically Workflow Nets (WF-nets) \cite{DBLP:journals/eor/SalimifardW01}, and Business Process Model and Notation (BPMN) \cite{DBLP:books/el/15/RosingWCM15} have gained widespread adoption, they suffer from limitations in terms of their ability to guarantee desirable quality properties such as \textit{soundness} (i.e., the absence of deadlocks and other anomalies).

The recently introduced Partially Ordered Workflow Language (POWL) \cite{powl} addresses these limitations by providing a powerful yet formally sound framework for process modeling. POWL is a hierarchical modeling language that allows for the construction of complex process models by combining smaller submodels using a set of well-defined operators. These operators include exclusive choice (XOR), loop, and partial order. POWL can be viewed as a generalization of process trees \cite{sander}, a notation widely used in various process mining techniques due to its quality guarantees. POWL preserves these desirable properties while increasing expressiveness through the partial order operator. Partial orders allow for the representation of activities that can be executed concurrently, but may have some ordering restrictions. \autoref{fig:ex:powl} illustrates an example POWL model, and \autoref{fig:ex:wf} shows a WF-net that captures the same behavior.

\begin{figure}[!t]
    \centering    
        \begin{subfigure}{0.35\textwidth}
            \centering
            \includegraphics[width=\textwidth]{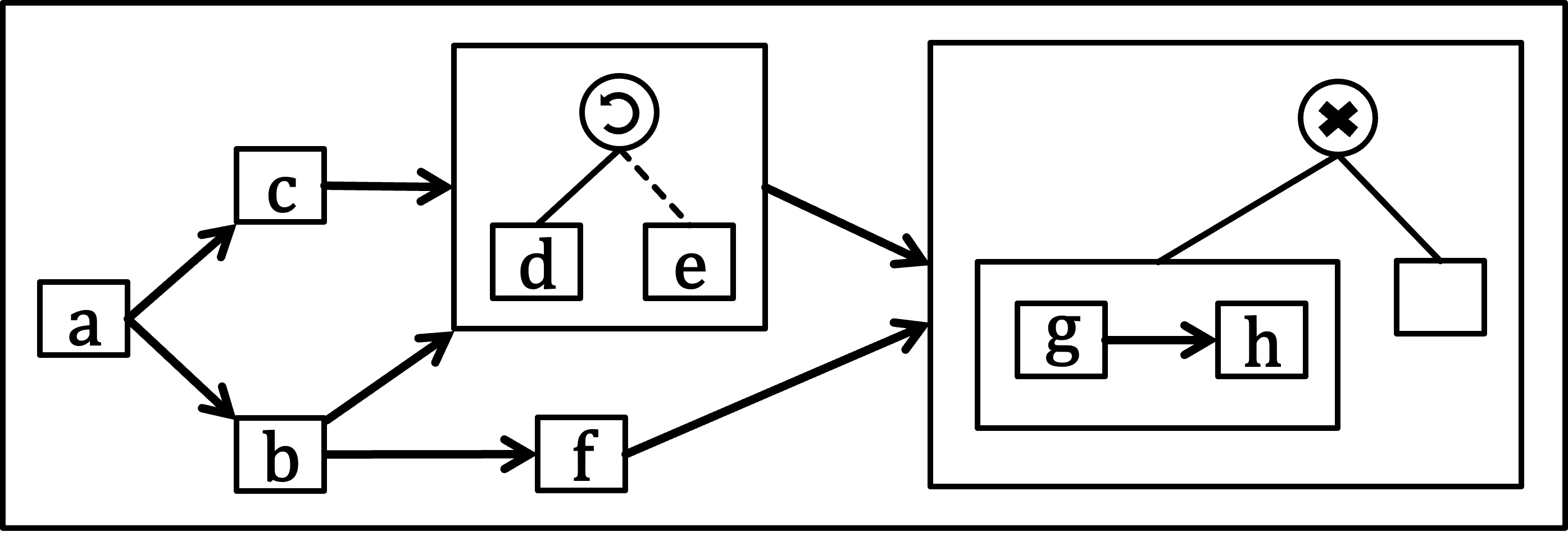}
            \caption{A POWL model.}\label{fig:ex:powl}
        \end{subfigure}
        \begin{subfigure}{0.09\textwidth}
            $ $
        \end{subfigure}
        \begin{subfigure}{0.54\textwidth}
            \centering
            \includegraphics[width=\textwidth]{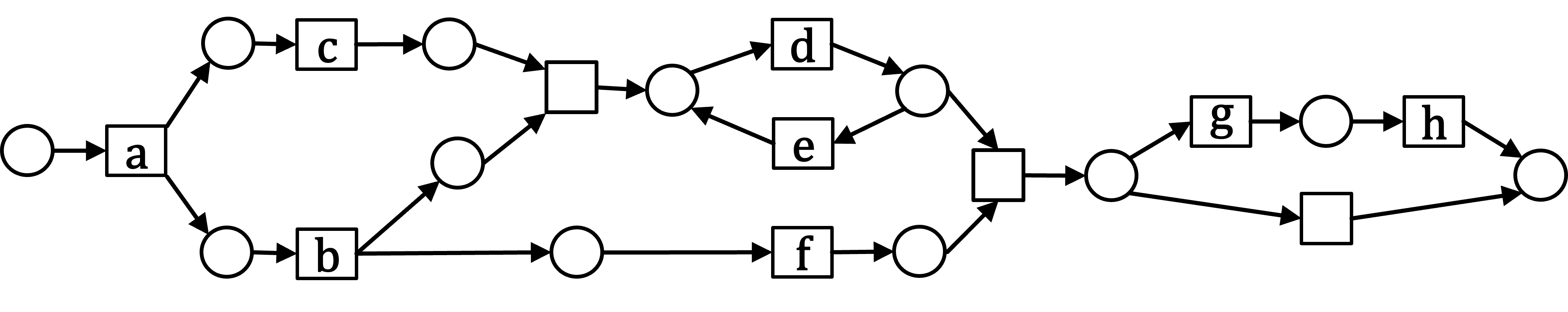}
            \caption{A WF-net.}\label{fig:ex:wf}
        \end{subfigure}

        \caption{Example process models.\label{fig:ex}}
\end{figure}

The hierarchical nature of POWL offers several advantages. The structured representation can significantly enhance the understandability of complex process models for humans, making it easier to grasp the overall control flow and identify potential areas for improvement. Furthermore, POWL opens up opportunities for developing faster and more efficient techniques for different process mining tasks, similar to how specialized algorithms have been optimized for process trees. Previous work has demonstrated the advantages of POWL in process discovery from data \cite{DBLP:conf/icpm/KouraniSA23,DBLP:journals/is/KouraniZSA25} and process modeling from text \cite{DBLP:conf/bpmds/KouraniB0A24}.

Despite its demonstrated benefits, the practical adoption of POWL is challenged by the prevalence of standard notations in existing tools and workflows. This motivates the need for a robust conversion from standard notations to POWL, allowing us to harness these benefits without requiring the adjustment of existing tools and practices. While the conversion from POWL to a sound WF-net is relatively straightforward, the inverse transformation presents a more significant challenge. This asymmetry arises from the fact that POWL represents a subclass of sound WF-nets. 

This paper proposes an algorithm that translates sound WF-nets into POWL. The proposed algorithm recursively decomposes the input WF-net into its smaller parts, identifies structural patterns corresponding to POWL's constructs, and assembles them into an equivalent POWL model. We formally prove the correctness of this transformation, demonstrating that the generated POWL model captures the language of the original WF-net. Furthermore, we define a subclass of WF-nets that encompasses equivalent representations for all POWL models, and we show the completeness of our approach on this subclass. We implement our algorithm and apply it on large WF-nets to demonstrates its high scalability. This work represents a significant step towards enabling the development of a process analysis and improvement techniques that utilize POWL internally while seamlessly accepting inputs in widely used formats.

The remainder of the paper is structured as follows. After discussing related work in \autoref{sec:rel}, we introduce necessary preliminaries in \autoref{sec:pre}. In \autoref{sec:convert}, we detail our algorithm for translating WF-nets into POWL. \autoref{sec:gua} formally proves the algorithm's correctness and completeness guarantees. Finally, we assess the scalabilty of our approach in \autoref{sec:eval}, and we conclude the paper in \autoref{sec:conc}.

\section{Related Work}\label{sec:rel}
Transformations between process modeling notations have been explored in various contexts. Some research focuses on transformations between different Petri net classes, such as the work on unfolding Colored Petri Nets into standard Place/Transition nets in \cite{DBLP:conf/apn/Dal-Zilio20} and the work on reducing free-choice Petri nets to either T-nets (also called marked graphs) or P-nets (also called state machines) in \cite{DBLP:conf/apn/Aalst21}. Other research addresses transformations between different types of process models. For example, \cite{DBLP:journals/infsof/DijkmanDO08} proposes an approach for converting BPMN models into Petri nets, \cite{gardner2003uml} discusses translating UML diagrams into BPEL, and \cite{DBLP:conf/apn/LangnerSW98} explores the mapping of Event-driven Process Chains (EPCs) into colored Petri nets. An overview on different approaches for the translation between workflow graphs and free-choice workflow nets is provided in \cite{DBLP:journals/is/FavreFV15}.


Transformations from graph-based formalisms like Petri nets into block-structured languages such as BPEL or process trees have been widely studied. The work on translating WF-nets to BPEL in \cite{DBLP:journals/infsof/AalstL08,DBLP:conf/otm/LassenA06} employs a bottom-up strategy, iteratively identifying patterns corresponding to BPEL fragments and substituting each identified pattern with a single transition to continue the recursion. This approach aims to maximize the size of the detected components in each iteration. The approach presented in \cite{DBLP:journals/algorithms/ZelstL20} for translating WF-nets into process trees, while also employing a bottom-up strategy, restricts the search space to patterns of size two. This approach cannot be adapted to POWL due to the presence of advanced partial order constructs that cannot be decomposed into components of size two. The fundamental difference between our algorithm and the aforementioned approaches, besides using different modeling languages, is that our approach employs a top-down strategy to ensure high scalability.

\section{Preliminaries}\label{sec:pre}
This section introduces fundamental preliminaries and notations.

\subsection{Basic Notations}\label{sec:pre:notation}

A \emph{multi-set} generalizes the notion of a set by tracking the frequencies of its elements. A multi-set over a set $X$ is expressed as $M = [{x_{1}}^{c_1},..., {x_{n}}^{c_n}]$ where $x_{1},..., x_{n} \in X$ are the elements of $M$ (denoted as $x_{i} \in M$ for $1\leq i\leq n$) and $M(x_i) = c_i \geq 1$ is the frequency of $x_{i}$ for $1\leq i\leq n$. 

A sequence of length $n \geq 0$ over a set $X$ is defined as function $\sigma{\colon}\{1,...,n\} \to X$, and we express it as $\sigma = \langle \sigma(1), ..., \sigma(n)\rangle$. The set of all sequences over $X$ is denoted by $X^{*}$. The concatenation of two sequences $\sigma_1$ and $\sigma_2$ is expressed as $\sigma_1 \cdot \sigma_2$, e.g., $\langle x_{1} \rangle \cdot \langle x_{2}, x_{1} \rangle = \langle x_{1}, x_{2}, x_{1} \rangle$. For two sets of sequences $L_1$ and $L_2$, we write $L_1 \cdot L_2 = \{\sigma_1 \cdot \sigma_2 \ | \ \sigma_1\in L_1 \ \wedge \ \sigma_2\in L_2\}$. 

Let ${\po \subseteq X\times X}$ be a binary relation over a set $X$. We use ${x_1 \po x_2}$ to denote ${(x_1,x_2)\in \po}$ and ${x_1 \notpo x_2}$ to denote ${(x_1,x_2) \notin \po}$. We define the \emph{transitive closure} of $\po$ as $\closure\po = \{(x, y) \mid \exists_{x_1, \dots, x_n \in X} \ x = x_1 \ \wedge \ y = x_n \ \wedge \forall_{1\leq i<n}, x_i \po x_{i+1}\}$.

A \emph{strict partial order} (\emph{partial order} for short) over a set $X$ is a binary relation that is \emph{irreflexive} ($x \notpo x$ for all ${x\in \X}$) and \emph{transitive} (${x_1 \po  x_2} \wedge {x_2 \po  x_3} \Rightarrow {x_1 \po  x_3}$). Irreflexivity and transitivity imply \emph{asymmetry} (${x_1 \po  x_2} \Rightarrow {x_2 \notpo  x_1}$). For $n \geq 2$, we use $\Orders{n}$ to denote the set of all partial orders over $\{1, \dots, n\}$. Let $X = \{x_1, ..., x_n\}$ be a set of size $n \geq 2$ and $\po \in \Orders{n}$. Then we write $\po(x_1, ..., x_n)$ to denote the partial order $\po'$ defined over $X$ as follows: 
$i \po j \Leftrightarrow x_i \po' x_j$ for all $i, j \in \{1, \dots, n\}$.


Let $\sigma_1, ..., \sigma_n \in X^*$ be $n \geq 2$ sequences over a set $X$ and $\po \in \Orders{n}$. The \emph{order-preserving shuffle operator} $\shuffle_{\po}$ generates the set of sequences resulting from interleaving $\sigma_1, ..., \sigma_n$ while preserving the partial order $\po$ of the sequences and the sequential order within each sequence. For example, let $\sigma_1 = \langle a, b\rangle$, $\sigma_2 = \langle c\rangle$, $\sigma_3 = \langle d, e\rangle$, and $\po = \{(1, 2), (1, 3)\} \in \Orders{3}$. Then, $\shuffle_{\po}(\sigma_1, \sigma_2, \sigma_3) = \{\langle a, b, c, d, e\rangle, \langle a, b, d, c, e\rangle, \langle a, b, d, e, c\rangle\}$.

For a set $X$, a \emph{partition} of $X$ of size $n \geq 1$ is a set of subsets $P = \{X_1, ... , X_n\}$ such that $X = X_1 \cup ... \cup X_n$, $X_i \neq \emptyset$, and $X_i \cap X_j = \emptyset$ for $1 \leq i < j \leq n$. For any $x  \in X$, we write $P_x$ to denote the subset of the partition (also called \emph{part}) that contains $x$, i.e., $P_x \in \{X_1, ... , X_n\}$ and $x \in P_x$. For example, let $P = \{\{a, b\}, \{c\}\}$ be a partition of $\{a, b, c\}$ of size $2$. Then, $P_a = P_b = \{a, b\}$ and $P_c = \{c\}$. 

\subsection{Workflow Nets}

We use $\ActUniverse$ to denote the set of all activities. We use $\tau \notin \ActUniverse$ to denote the \emph{silent activity}, which is used to model a choice between executing or skipping a path in process model, for example. To enable creating process models with duplicated activities, we introduce the notion of \emph{transitions}, and we use $\TraUniverse$ to denote the set of all transitions. Each transition is mapped to an activity, denoted as the $label$ of the transition. We use $\lab{\colon} \TraUniverse \to \ActUniverse \cup \{\tau\}$ to denote the labeling function.

A \textit{Petri net} is a directed bipartite graph consisting of two types of nodes: \textit{places} and \textit{transitions}. Transitions represent instances of activities, while places are used to model dependencies between transitions.

\begin{definition}[Petri Net]\label{def:petrinet}
A Petri net is a triple $N = (P, T, F)$, where $T \subset \TraUniverse$ is a finite set of transitions, $P$ is a finite set of places such that $T \cap P = \emptyset$, and $F \subseteq (P \times T) \cup (T \times P)$ is the flow relation.  
\end{definition}

Let $N=(P,T, F)$ be a Petri net. We define the following notations: 
\begin{itemize}
    \item For $x \in P \cup T$, $\pre{x} = \{y \ | \ (y, x) \in F\}$ is the \emph{pre-set} of $x$, and $\post{x} = \{y \ | \ (x, y) \in F\}$ is the \emph{post-set} of $x$.
    \item For $T' \subseteq T$, we define the \emph{projection} of $P$ on $T'$ as $P\project{T'} = \{ p \in P \ | \ (\pre{p} \cup \post{p}) \cap T' \neq \emptyset\}$.
    \item For $P' \subseteq P$ and $T' \subseteq T$, we define the \emph{projection} of $F$ on $P$ and $T$ as $F\project{P',T'} = F \cap ((P' \times T') \cup (T' \times P'))$.
    \item For $T' \subseteq T$, two places $p,p' \in P$ are \emph{equivalent with respect to $T'$}, denoted as $p\approx_{T'} p'$, iff $(\pre{p} \cap T' = \pre{p'} \cap T') \ \wedge \ (\post{p} \cap T' = \post{p'} \cap T')$.
\end{itemize}

Places hold \emph{tokens}, and a transition is considered \emph{enabled} if each of its preceding places has at least one token. \emph{Firing} an enabled transition consumes one token from each of its preceding places and produces a token in each of its succeeding places. A \emph{marking} is a multi-set of places indicating the number of tokens in each place. A Petri net is called \emph{safe} if each place in the net cannot hold more than one token. Next, we define three subclasses of Petri nets. More details on these classes can be found in \cite{DBLP:conf/ac/1996petri1,desel1995free,DBLP:journals/eor/SalimifardW01}.


\begin{definition}[Marked Graph, Free-Choiceness, Workflow Net]\label{def:petri-sub}
Let $N=(P,T, F)$ be a Petri net. Then, the following holds:
\begin{itemize}
    \item $N$ is a marked graph iff for any $p \in P$: $|\pre{p}| \leq 1 \ \wedge \ |\post{p}| \leq 1$.  
    \item $N$ is free-choice iff for any $t_1, t_2 \in T$: $(\pre{t_1} \cap \pre{t_2} \neq \emptyset) \Rightarrow (\pre{t_1} = \pre{t_2})$. 
    \item $N$ is a workflow net (WF-net) iff places $N_{source}, N_{sink}\in P$ exist such that:
    \begin{itemize}
        \item \textbf{Unique source:} $\{N_{source}\} = \{p \in P \ | \ \pre{p} = \emptyset\}$.
        \item \textbf{Unique sink:} $\{N_{sink}\} = \{p \in P \ | \ \post{p} = \emptyset\}$.
        \item \textbf{Connectivity:} each node is on a path from $N_{source}$ to $N_{sink}$.
    \end{itemize}
\end{itemize}
\end{definition}





Let $N=(P,T, F)$ be a WF-net. We use $N_{source}$ to denote the unique source place and $N_{sink}$ to denote the unique sink place. Furthermore, we define: 
\begin{itemize}
    \item For $T' \subseteq T$, $\Pre{T'} = \{p \in P \ | \ T' \cap \post{p} \neq \emptyset \ \wedge \ (p = N_{source} \ \vee \ (T \setminus T') \cap \pre{p} \neq \emptyset)\}$ is the set of \emph{entry points} of $T'$.
    \item For $T' \subseteq T$, $\Post{T'} = \{p \in P \ | \ T' \cap \pre{p} \neq \emptyset \ \wedge \ (p = N_{sink} \ \vee \ (T \setminus T') \cap \post{p} \neq \emptyset)\}$ is the set of \emph{exit points} of $T'$.
    \item For a partition $G = \{T_1, \dots, T_n\}$ of $T$, the \emph{execution order} of $G$ within $N$ is the binary relation $\mathit{order}(N, G) = \{(i, j) \ | \ i, j \in \{1, \dots, n\} \ \wedge \ (\Post{T_i} \cap \Pre{T_j}) \neq \emptyset\}$.
\end{itemize}

WF-nets may suffer from quality anomalies (e.g., transitions that can never be enabled). WF-nets without such undesirable properties are called \emph{sound}.

\begin{definition}[Soundness]\label{def:sound}
Let $N=(P,T, F)$ be a WF-net. $N$ is sound iff the following conditions hold:
\begin{itemize}
    \item \textbf{No dead transitions:} for each transition $t\in T$, there exists a marking $M$ reachable from $[N_{source}]$ that enables $t$.
    \item \textbf{Option to complete:} for every marking $M$ reachable from $[N_{source}]$, there exists a firing sequence leading from $M$ to $[N_{sink}]$. 
    \item \textbf{Proper completion:} $[N_{sink}]$ is the only marking reachable from $[N_{source}]$ with at least one token in $N_{sink}$. 
\end{itemize}
\end{definition}

The WF-net shown in \autoref{fig:ex:wf} is sound. Note that the option to complete implies proper completion.

\subsection{POWL Language}\label{sec:powl}
A POWL model is constructed recursively from a set of activities, combined either as partial orders or using the control flow operators $\xor$ and $\Loop$. The operator $\xor$ models an exclusive choice between submodels, while $\Loop$ model cyclic behavior between two submodels: the \textit{do-part} is executed first, and each time the \textit{redo-part} is executed, it is followed by another execution of the do-part. In a partial order, all submodels are executed, while respecting the given execution order.

POWL models are defined in \cite{powl} over activities. We redefine POWL models over transitions to allow for models with multiple instances of the same activity. 

\begin{definition}[POWL Model]\label{def:powl} 
POWL models are defined as follows:
\begin{itemize}
    \item Any transition ${t \in \TraUniverse}$ is a POWL model.
    \item Let $\powl_1, ..., \powl_n$ be $n \geq 2$ POWL models.
    \begin{itemize}
        \item $\xor(\powl_1, ..., \powl_n)$ is a POWL model.
        \item $\Loop(\powl_1, \powl_2)$ is a POWL model.
        \item For any partial order $\po \in \Orders{n}$, $\po(\powl_1, ..., \powl_n)$ is a POWL model.
    \end{itemize}
\end{itemize}
\end{definition}

\autoref{fig:ex:powl} shows an example POWL model. The language a POWL model is defined recursively based on the semantics of its operators.

\begin{definition}[POWL Semantics]\label{def:lang}
The language of a POWL model $\powl$ is recursively defined as follows:
\begin{itemize}
    \item ${\lang(t) = \{\langle a \rangle\}}$ for $t \in \TraUniverse$ with $\lab(t) = a \in \ActUniverse$.
    \item ${\lang(t) = \{\langle \rangle\}}$ for $t \in \TraUniverse$ with $\lab(t) = \tau$.
    \item Let $\powl_1, ..., \powl_n$ be ${n \geq 2}$ POWL models with $L_i = \lang(\powl_i)$ for $1 \leq i \leq n$.
    \begin{itemize}
        \item $\lang(\xor(\powl_1, ..., \powl_n)) = \bigcup\limits_{1\leq i \leq n} {L_i}$.
        \item $\lang(\Loop(\powl_1, \powl_2)) = L_1 \cdot (L_2 \cdot L_1)^*$.
         \item For $\po \in \Orders{n}$, $\lang(\po(\powl_1, ..., \powl_n)) = \{\sigma \in \shuffle_{\po}(\sigma_1 , ..., \sigma_n) \ | \ \forall_{1 \leq i \leq n} \ \sigma_i \in L_i \}$.           
    \end{itemize}
\end{itemize}
\end{definition}

\subsection{Semi-Block-Structured WF-nets}
POWL is less expressive than WF-nets, meaning that not all WF-nets have equivalent POWL models. To clearly define the scope of our WF-net to POWL translation algorithm, we extend the concept of \emph{block-structured} workflow nets defined in \cite{sander}. A block-structured WF-net is a sound WF-net that can be divided recursively into parts having single entry and exit points. In other words, there must be a unique mapping between every place/transition with multiple outgoing arcs and every place/transition with multiple incoming arcs to mark the start and end of a block.

POWL can represent all block-structured WF-nets, and the introduction of the partial order operator in POWL extends its expressiveness beyond this class. We can relax the block-structure requirements since the partial order operator allows for the concurrent executions of transitions without imposing block structure on them. In other words, the block structure is only required for decision points (i.e., places) due to the usage of the process tree operators $\xor$ and $\Loop$.


\begin{definition}[Semi-Block-Structured WF-nets]\label{def:block}
Let $N = (P, T, F)$ be a WF-net. $N$ is \emph{semi-block-structured} iff it satisfies the following conditions:

\begin{itemize}
    \item $N$ is sound and safe.

    \item \textbf{Explicit decision points:} For each $(x, y) \in F: \card{\pre{y}} = 1 \vee \card{\post{x}} = 1$.
    
    \item \textbf{Unique mapping between split and join decision points:} Let
    $P_{\text{split}} = \{ p \in P \mid |\post{p}| > 1 \}$ and
    $P_{\text{join}} = \{ p \in P \mid |\pre{p}| > 1 \}$. There exists a bijective mapping $
    \mathcal{B}: P_{\text{split}} \rightarrow P_{\text{join}}$ such that for each $(p, p') \in \mathcal{B}$: $|\post{p}| = |\pre{p'}|$.
    
    \item \textbf{Disjoint subnets between decision points:} For each pair $(p, p') \in \mathcal{B}$, let $k = |\post{p}| = |\pre{p'}|$. There exist $k$ disjoint, non-empty sets of transitions $T_1, \dotsc, T_k \subseteq T$ such that for each i ($1\leq i \leq k$), $P_i = P|_{T_i}$, and $F_i = F|_{P_i, T_i}$:
    \begin{itemize}
        \item $P_i \cap P|_{T\setminus T_i} = \{ p, p' \}$.
        \item $N_i = (P_i, T_i, F_i)$ is a workflow net with $\{ N_{i_{source}}, N_{i_{sink}} \} = \{ p, p' \}$.
    \end{itemize}
    
    
\end{itemize}

\end{definition}

\begin{figure}[!t]
    \centering    
        \includegraphics[width=0.7\textwidth]{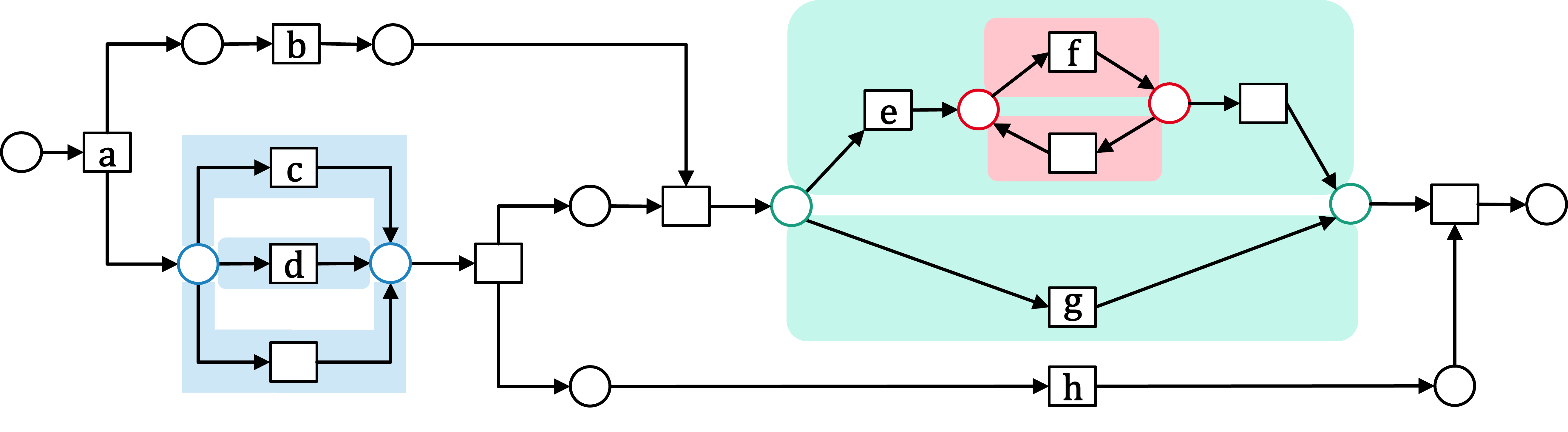}
        \caption{A semi-block-structured WF-net with three blocks, highlighted in different colors.\label{fig:ex:blockWF}}
\end{figure}

The WF-net in \autoref{fig:ex:blockWF} is semi-block-structured WF-net but not block structured. Note that a semi-block-structured WF-net with no blocks is a marked graph. Furthermore, a semi-block-structured WF-net is free-choice due to the explicit decision points requirement. Therefore, we can conclude that semi-block-structured WF-nets are a subset of sound free-choice WF-nets, a superset of block-structured WF-nets, and superset of sound marked graph WF-nets. 

It is trivial to observe that, for any POWL model, there exists at least one equivalent semi-block-structured WF-net. Furthermore, after introducing our conversion algorithm, we will show in \autoref{sec:gua:redisc} its completeness on this class of WF-nets; i.e., we will show the our algorithm successfully converts any semi-block-structured WF-net into a POWL model that captures the same language.


\section{Transforming Workflow Nets into POWL}\label{sec:convert}
This section presents a recursive algorithm for transforming safe and sound WF-nets into equivalent POWL models. First, the WF-net can be preprocessed by applying a set of reduction rules. Then, the algorithm checks whether the WF-net forms a structural pattern that corresponds to a POWL component (i.e., choice, loop, or partial order). The identified pattern is then translated into its corresponding POWL representation, and the WF-net is projected onto subsets of transitions to continue the recursion.

\subsection{Preprocessing}


To expand the applicability of our algorithm, the input WF-net can be preprocessed by applying a series of reduction rules. This is an optional step, aiming at bringing the WF-net closer to the structure expected in the subsequent steps of the algorithm. 
Numerous reduction rules have been proposed in the literature, such as those presented in \cite{desel1995free,DBLP:conf/ac/Berthelot86,10.1007/BFb0016204,DBLP:journals/pieee/Murata89}. Any reduction rule can be applied as long as it preserves the essential structural properties of the WF-net, namely its language, safeness, and soundness. In \autoref{fig:prepro}, we illustrate reduction rules for introducing explicit places for decision points. Additional reduction rules are introduced in \autoref{sec:loop} to enable the detection of special loop structures.


\begin{figure}[!t]
    \centering    
        
        \begin{subfigure}{0.48\textwidth}
            \centering
            \includegraphics[width=0.8\textwidth]{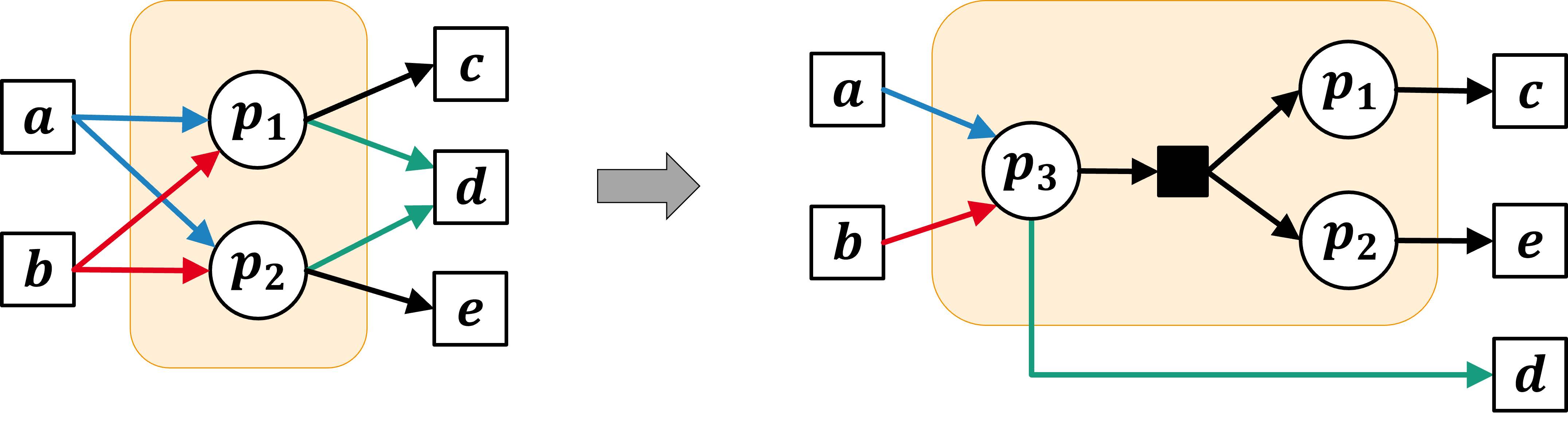}
            \caption{Introducing explicit XOR split places.}
        \end{subfigure}
        \begin{subfigure}{0.02\textwidth}
                $ $
        \end{subfigure}
        \begin{subfigure}{0.48\textwidth}
            \centering
            \includegraphics[width=0.8\textwidth]{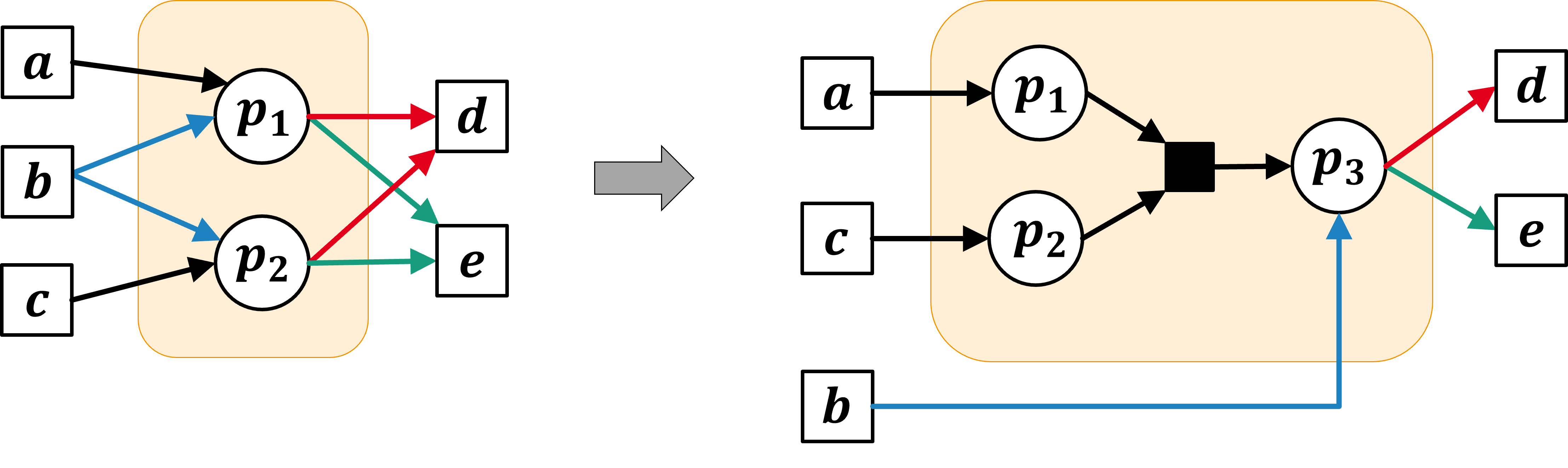}
            \caption{Introducing explicit XOR join places.}
        \end{subfigure}

        \caption{Reduction rules illustrated with examples.\label{fig:prepro}}
\end{figure}

\subsection{Identifying Choices}
This section addresses partitioning a WF-net into smaller subnets such that the given WF-net corresponds to an XOR structure over the identified parts. 

We first introduce the concept of \textit{transition reachability}, capturing the notion of one transition being reachable from another through a path in the WF-net.

\begin{definition}[Transition Reachability]
Let $N = (P, T, F)$ be a Petri net. The \emph{transition reachability relation} $\TTR \subseteq T \times T$ is defined as follows for $t, t' \in T$:
\[
t \TTR t' \iff \exists_{p_1, \dots, p_n \in P} \text{ and }  \exists_{t_1, \dots, t_{n+1} \in T} 
\]
such that $t_1 = t$, $t_{n+1} = t'$, and for each $i$ ($1 \leq i \leq n$):
\[
(t_i, p_{i}) \in F \quad \text{and} \quad (p_{i}, t_{i+1}) \in F.
\]

\end{definition}

An \emph{XOR pattern}, as illustrated in \autoref{fig:xor}, represents a WF-net where transitions can be partitioned such that exactly one part can be executed in a single process instance. Intuitively, transitions belonging to different parts cannot be reachable from each other.

\begin{figure}[!t]
    \centering    
     \includegraphics[width=0.5\textwidth]{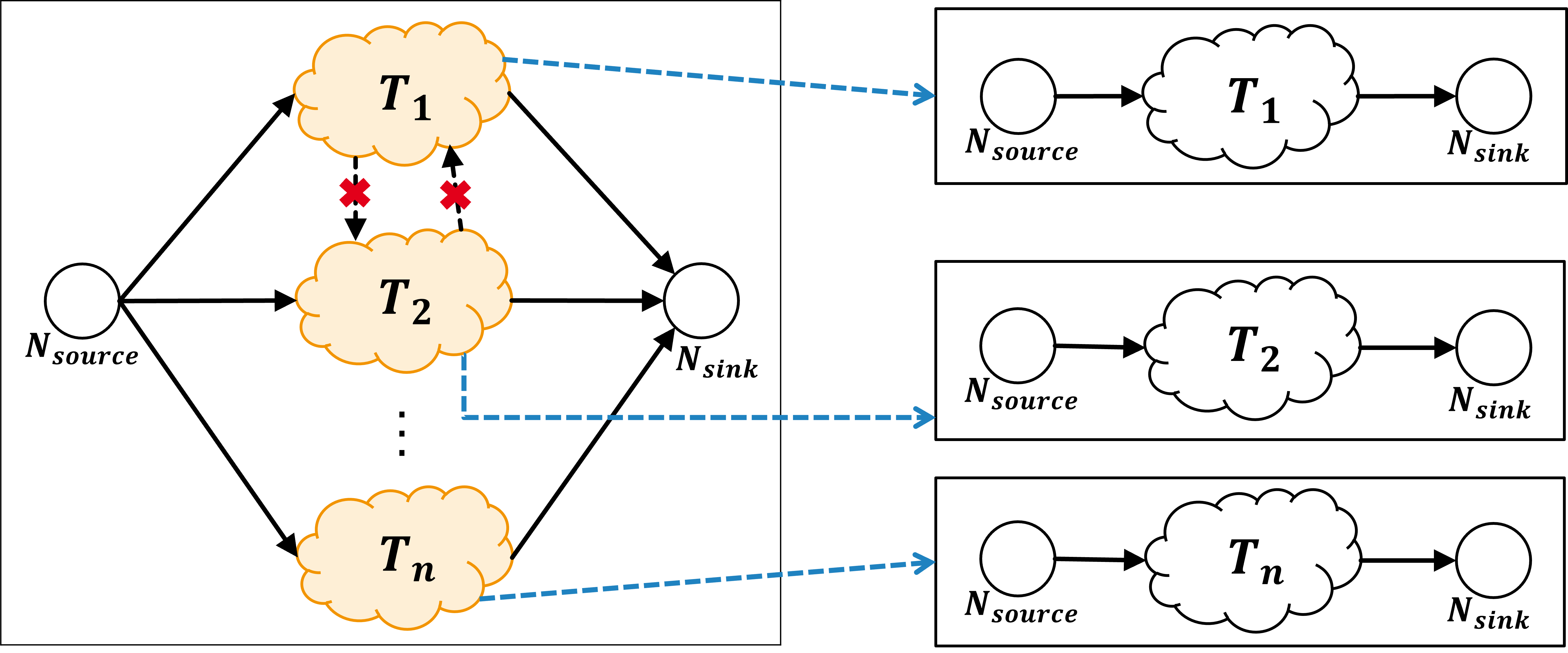}
    \caption{XOR pattern and projection.\label{fig:xor}}
\end{figure}

\begin{definition}[XOR Pattern]\label{def:pattern_xor}
Let $N = (P, T, F)$ be a safe and sound WF-net. Let $G = \{T_1, \ldots, T_n\}$ be a partition of transitions of size $n \geq 2$. The tuple $(N, G)$ is called an XOR pattern iff for all $(t, t') \in T \times T$:
\[
\text{if } t \TTR t' \text{, then } G_t = {G_{t'}}.
\]
 \end{definition}

For a WF-net $N$, we use $\xorpart(N)$ to denote the partition generated by iteratively grouping transitions based on their reachability relation, aligning with \autoref{def:pattern_xor}. Note that $(N, \xorpart(N))$ is an XOR pattern if $\card{\xorpart(N)} \geq 2$. 

After identifying an XOR pattern, the WF-net is projected onto the different parts, creating several subnets for the recursive application of the algorithm. The projection is achieved by selecting the relevant places and flow relations, isolating the chosen part. 

\begin{definition}[XOR Projection]\label{def:projection_xor}
Let $(N, G)$ be an XOR pattern with $N = (P, T, F)$ and $T' \in G$ be a part. The XOR projection of $N$ on $T'$ is defined as $\xorproject(N, T') = (P', T', F')$ where $P' = P\project{T'}$ and $F' = F\project{P', T'}$.
\end{definition}



            
            

\subsection{Identifying Loops} \label{sec:loop}
This section addresses partitioning a WF-net into subnets such that the given WF-net corresponds to a loop structure over the identified do- and redo-parts. 

We first define the concept of \textit{in-between places reachability} to identify the transitions that lie on paths between two specific places in a WF-net. 


\begin{definition}[In-Between Places Reachability]
Let $N = (P, T, F)$ be a Petri net. The set of reachable transitions between two places $p, p' \in P$, denoted as $\PTR(p, p')$, is defined as follows:
\[
t \in \PTR(p, p') \iff \exists t_1, \dots, t_{n} \in T \text{ and } p_1, \dots, p_{n+1} \in P \]
such that
\[ 
p_1 = p \ \wedge \ p_{n+1} = p' \ \wedge \ t \in \{t_1, \dots, t_{n}\},
\]
and for each $i$ ($1 \leq i \leq n$):
\[
p_i \neq p' \ \wedge \ (p_i, t_i) \in F \ \wedge \ (t_i, p_{i+1}) \in F.
\]
\end{definition}

\begin{figure}[!t]
    \centering    
     \includegraphics[width=0.7\textwidth]{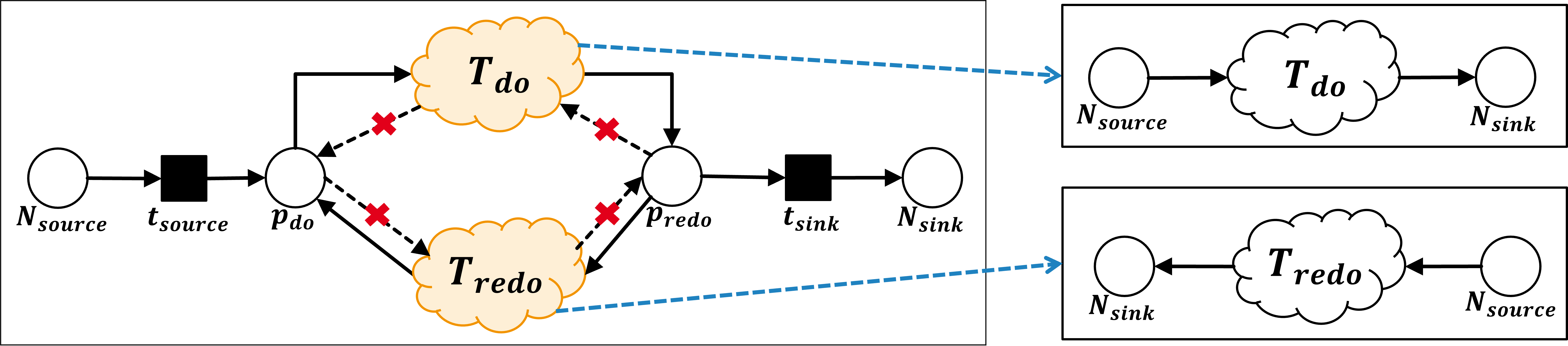}
    \caption{Loop pattern and projection.\label{fig:loop}}
\end{figure}

An \emph{loop pattern}, as illustrated in \autoref{fig:loop}, represents a WF-net where transitions can be partitioned into three parts: do-part, redo-part, and silent transitions for loop entry and exit. A loop pattern must also include two special places, $p_{do}$ and $p_{redo}$, which mark the entry and exit points of the loop. The do-part consists of transitions reachable from $p_{do}$ to $p_{redo}$, while the redo-part consists of transitions reachable from $p_{redo}$ to $p_{do}$.

\begin{definition}[Loop Pattern]\label{def:loop_pattern}
Let $N = (P, T, F)$ be a safe and sound WF-net. Let $G = \{T_{do}, T_{redo}, T_{\tau}\}$ be a partition of transitions of size $3$. The tuple  $(N, G)$ is called a loop pattern iff places $p_{do}, p_{redo} \in P$ exist such that: 
    \begin{enumerate}
        \item Two transitions $t_{source}$ and $t_{sink} $ exist such that $t_{source} \neq t_{sink}$, $T_{\tau} = \{t_{source}, t_{sink}\}$, and $\lang(t_{source}) = \lang(t_{sink}) = \tau$.
        \item $\post{N_{source}} = \{t_{source}\}$ and $\pre{N_{sink}} = \{t_{sink}\}$.
        \item $\pre{t_{source}} = \{N_{source}\}$ and $\post{t_{source}} = \{p_{do}\}$.
        \item $\pre{t_{sink}} = \{p_{redo}\}$ and $\post{t_{sink}} = \{N_{sink}\}$.                     
        
        \item $T_{do} = \PTR(p_{do}, p_{redo})$ and $T_{redo} = \PTR(p_{redo}, p_{do})$.  

        \item $\pre{p_{do}} \cap T_{do} = \emptyset$ and $\pre{p_{redo}} \cap T_{redo} = \emptyset$. 

        \item $\post{p_{redo}} \cap T_{do} = \emptyset$ and $\post{p_{do}} \cap T_{redo} = \emptyset$. 
        
    \end{enumerate}
\end{definition}

Note that a WF-net satisfying the requirements \textit{1 - 5} of \autoref{def:loop_pattern} can be preprocessed, as illustrated in \autoref{fig:loop_prepro:def_req}, to satisfy \textit{6} without affecting its language, soundness, or safeness. Furthermore, the requirement \textit{7} is implied by \textit{1 - 5}.

\begin{figure}[!t]
    \centering    
        \begin{subfigure}{0.5\textwidth}
            \centering
            \includegraphics[width=0.8\textwidth]{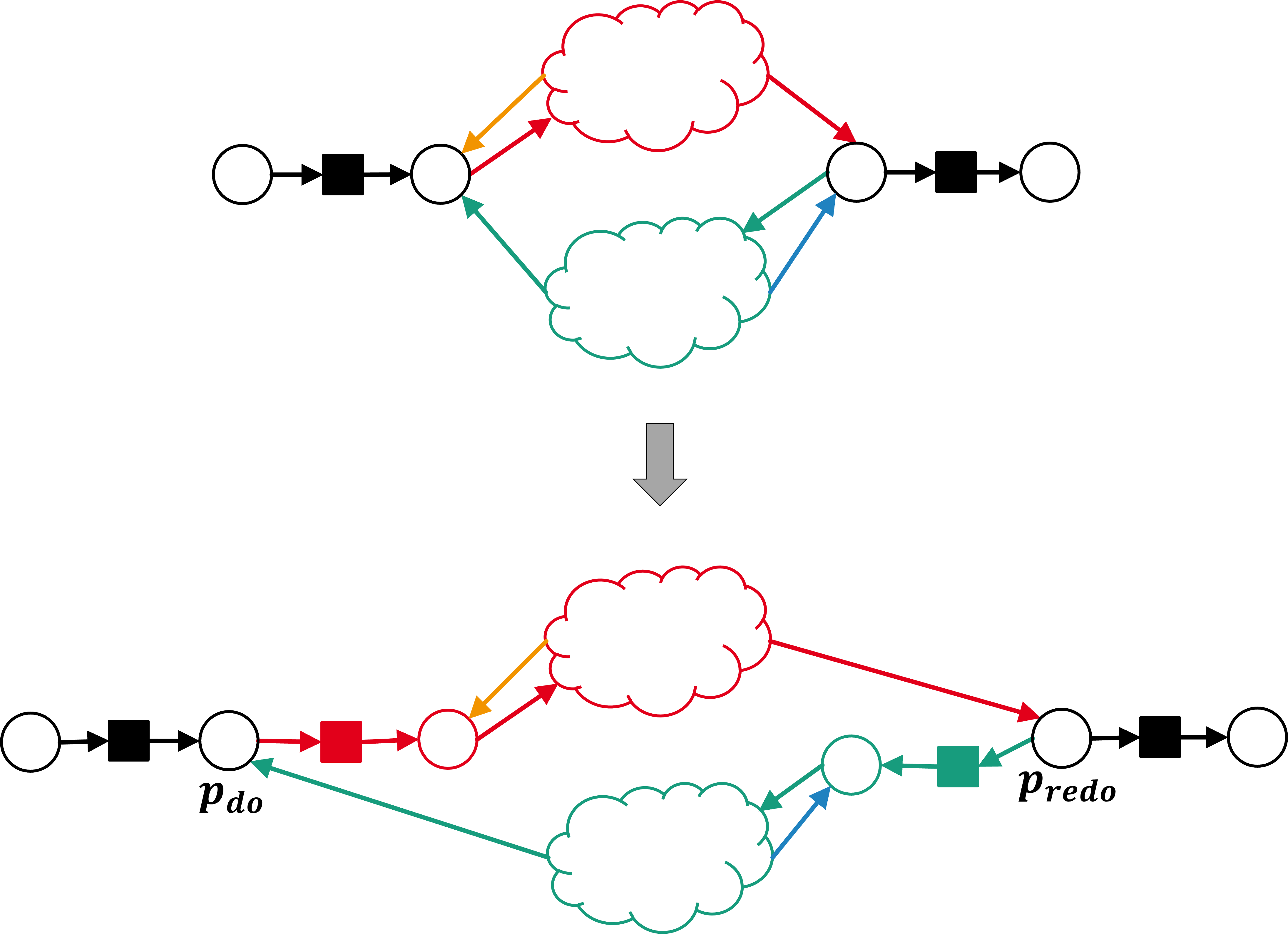}
            \caption{Satisfying the requirement \textit{6} of \autoref{def:loop_pattern}.\label{fig:loop_prepro:def_req}}
        \end{subfigure}
        \begin{subfigure}{0.13\textwidth}
        $ $
        \end{subfigure}
        \begin{subfigure}{0.35\textwidth}
            \centering
            \includegraphics[width=0.8\textwidth]{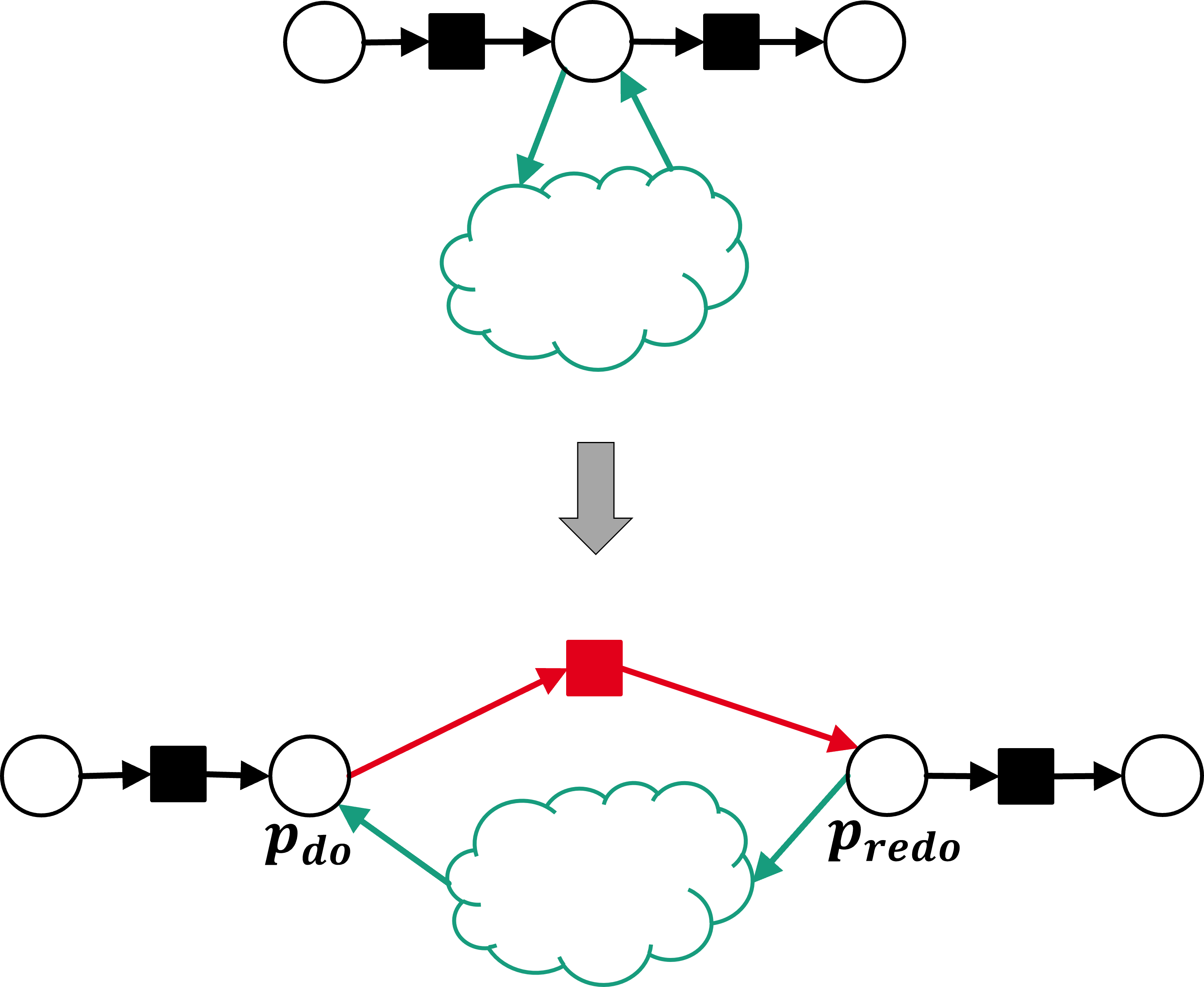}
            \caption{Handling self-loops. \label{fig:loop_prepro:self_loop}}
        \end{subfigure}

        \caption{Reduction rules for enabling loop pattern detection.\label{fig:loop_prepro}}
\end{figure}

\paragraph{Identifying Self-Loops with Missing Do-Part.}
A WF-net may represent a self-loop with a single place marking both the start and end of the loop. In such cases, all activities belong to the redo-part, which can be skipped or repeated. To enable the detection of these loops, we restructure the WF-net, as illustrated in \autoref{fig:loop_prepro:self_loop}, by adding a silent transition to represent the do-part.

After detecting a loop pattern, the WF-net is projected into the do-part and redo-parts, creating two subnets, and the recursion continues on the created subnets. As illustrated in \autoref{fig:loop}, the loop projection is done by selecting the appropriate subset of transitions and adjusting the flow relation to correctly connect the subnet to the source and sink place.

\begin{definition}[Loop Projection]\label{def:projection_loop}
Let $(N, G)$ be an loop pattern with $N = (P, T, F)$; $T_{do}, T_{redo} \in G$; and $p_{do}, p_{redo} \in P$ as defined in \autoref{def:loop_pattern}. The loop projection of $N$ on $T' \in \{ T_{do}, T_{redo}\}$ is $\loopproject(N, T') = (P', T', F')$ where
\[
P' = (P\project{T'} \setminus \{p_{do}, p_{redo}\}) \cup \{N_{source}, N_{sink}\}
\] and
\[
F' = F\project{P', T'} \ \cup \ \{(N_{source}, t) \ | \ t \in T' \ \wedge \ (p_{start}, t) \in F \} 
\] \vspace{-12pt} 
\[
\cup \ \{(t, N_{sink}) \ | \ t \in T' \ \wedge \ (t, p_{end}) \in F \}
\] with 
\[
(p_{start}, p_{end}) =
    \begin{cases}
        (p_{do}, p_{redo}) & \text{if } T' = T_{do}, \\
        (p_{redo}, p_{do}) & \text{if } T' = T_{redo}.
    \end{cases}
\]
\end{definition}

\subsection{Identifying Partial Orders}
This section addresses partitioning a WF-net into smaller subnets such that the given WF-net corresponds to a partial order over the identified parts.

A \emph{partial order pattern} represents a WF-net with a partition  of transitions such that all parts are executed and the execution order forms a partial order.

\begin{definition}[Partial Order Pattern]\label{def:po_pattern}
Let $N = (P, T, F)$ be a safe and sound WF-net. Let $G = \{T_1, \dots, T_n\}$ be a partition of transitions of size $n \geq 2$. The tuple $(N, G)$ is called a partial order pattern iff the following conditions hold: 
\begin{enumerate}
    \item For all $t, t' \in T$ and $p \in P$:
    \[
        \text{if } t, t' \in \{t \in T \ | \ \exists_{t_1, t_2 \in \post{p}} \ t_1 \TTR t \wedge t_2 \notTTR t \}\text{, then } G_{t} = G_{t'}.
    \]
    \item $\po = \closure{\mathit{order}(N, G)}$ is a partial order.

    \item \textbf{Unique local start:} for all $i \in \{1, \ldots, n\}$ and $p, p' \in \Pre{T_i}$: $p \approx_{T_i} p'$.

    \item \textbf{Unique local end:} for all $i \in \{1, \ldots, n\}$ and $p, p' \in \Post{T_i}$: $p \approx_{T_i} p'$.

\end{enumerate}
\end{definition} 

\begin{figure}[!t]
    \centering    
     \includegraphics[width=0.9\textwidth]{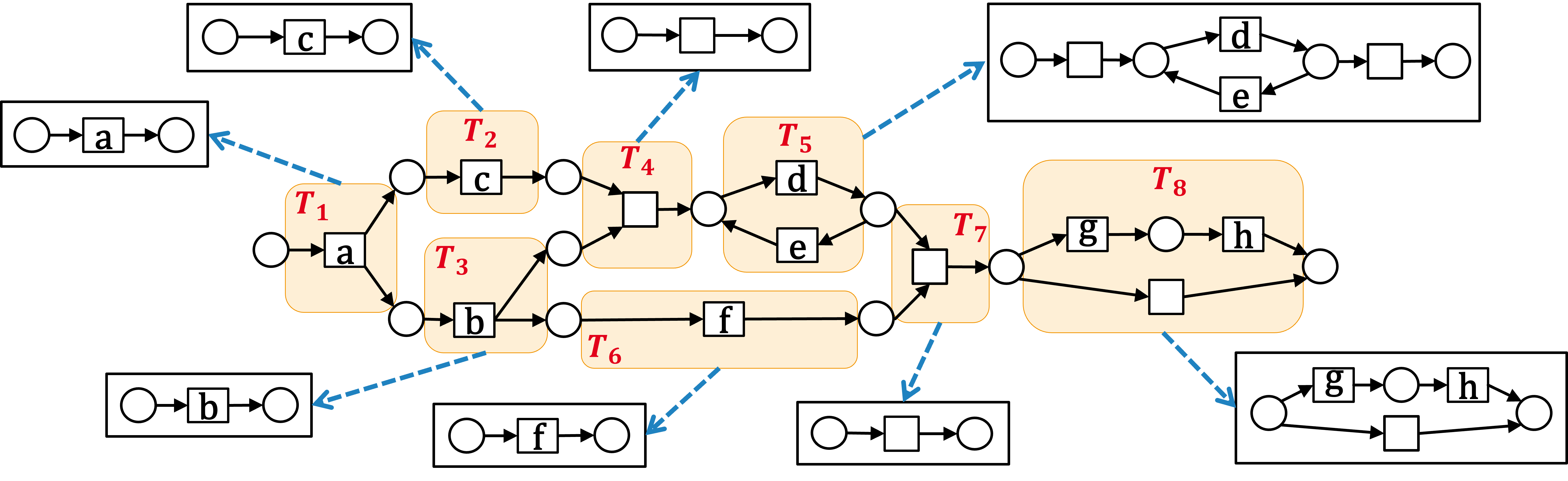}
    \caption{An example illustrating a partial order pattern $(N, \{T_1, \dots, T_8\})$ and the partial order projection of the WF-net $N$ on the identified parts.\label{fig:order}}
\end{figure}

The first condition of \autoref{def:po_pattern} states that if a place has outgoing flows leading to different transitions, then all such transitions must fall into the same part, as the place represents a decision point (i.e., there is a potential choice or loop within this part). The second condition requires that the transitive closure of the execution order between the parts of the partition forms a partial order. The third and fourth conditions ensure that the identified parts form cleanly separable components that can be executed independently with well-defined entry and exit points. \autoref{fig:order} shows a partial order pattern detected for the WF-net from \autoref{fig:ex:wf}.

For a WF-net $N$, we use $\popart(N)$ to denote the partition generated by iteratively grouping transitions based on their reachability relations with respect to decision points, aligning with the first condition of \autoref{def:po_pattern}. Note that $(N, \popart(N))$ is a partial order pattern if $\card{\popart(N)} \geq 2$ and the remaining three requirements of \autoref{def:po_pattern} are met.

\paragraph{Normalization.} Before defining the projection for partial order patterns, we introduce the concept of normalization. Let $N$ be a Petri net with known unique start and end places $p_s \in P$ and $p_e \in P$, respectively. Then $P$ is normalized into a new Petri net $\mathit{Normalize}(N, p_s, p_e)$ by (i) inserting a new start place and connecting it to $p_s$ through a silent transition in case $\pre{p_s} \neq \emptyset$ and (ii) inserting a new end place and connecting it to $p_e$ through a silent transition in case $\post{p_e} \neq \emptyset$. Normalization aims at ensuring conformance with the requirements of WF-nets (c.f. \autoref{def:petri-sub}) by adding new source and sink places if needed.

After detecting a partial order pattern, the WF-net is projected on the identified parts as illustrated in the example shown in \autoref{fig:order}. This projection is done by selecting the appropriate subset of transitions, adding unique start and end places to represent the entry and exit points of the part, adjusting the flow relation accordingly, and applying normalization if needed.

\begin{definition}[Partial Order Projection]\label{def:projection_po}
Let $(N, G)$ be a partial order pattern with $N = (P, T, F)$. Let $T' \in G$ be a part. The partial order projection of $N$ on $T'$ is $\poproject(N, T') = \mathit{Normalize}(N', p_s, p_e)$ where $p_s, p_e \notin P$ are two fresh places and $N' = (P', T', F')$ is constructed as follows:
\begin{itemize}
    \item $P' = (P\project{T'} \setminus (\Pre{T'} \cup \Post{T'})) \cup \{p_s, p_e\}$.
    \item $F' = F\project{P', T'}$\\ $\cup \ \{(p_s, t) \ | \ (p, t) \in F \text{ for } p\in\Pre{T'}\} \ \cup \ \{(t, p_s) \ | \ (t, p) \in F \text{ for } p\in\Pre{T'}\}$\\
    $ \cup \ \{(p_e, t) \ | \ (p, t) \in F \text{ for } p\in\Post{T'}\} \ \cup \ \{(t, p_e) \ | \ (t, p) \in F \text{ for } p\in\Post{T'}\}$.

\end{itemize}
\end{definition}


    
        


\subsection{WF-Net to POWL Converter}
\autoref{alg:convert_pn_to_powl} converts a safe and sound WF-nets into an equivalent POWL model. 
First, the algorithm checks whether the WF-net consists of a single transition (base case). If a base case is not detected, the algorithm attempts to identify an XOR, loop, or partial order pattern. If a pattern is found, the algorithm projects the WF-net on the identified parts, recursively converts the created subnets into POWL models, and combines them using the appropriate POWL operator. If no pattern is detected, the algorithm returns $null$, indicating that the WF-net cannot be converted into a POWL model.

\begin{algorithm}[!t]
\smaller
\caption{Conversion of a WF-Net into a POWL Model.}\label{alg:convert_pn_to_powl}
\DontPrintSemicolon
\KwIn{A safe and sound WF-net $N = (P, T, F)$.}
\KwOut{A POWL model or $null$ if no translation is possible.}

\SetKwFunction{FMain}{ConvertNetToPOWL}
\SetKwProg{Fn}{Function}{:}{}
\Fn{\FMain{$N$}}{

    
        \If{$\card{T} = 1$ with $T = \{t\}$, $\card{P} = 2$, and $F = \{(N_{source}, t), (t, N_{sink})\}$}{
            \Return $t$ \;
        }

        $G \leftarrow \xorpart(N)$\;
        
        \If{$(N, G)$ is an XOR pattern}{
            \For{$T_i \in G = \{T_1, \dots, T_n\}$}{
                $\powl_i \leftarrow$ \FMain{$\xorproject(N, T_i)$}\;
            }
            \Return $\xor(\powl_1, \dots, \powl_n)$\;
        }

     
        \If{a loop pattern exists with $p_{do}, p_{redo} \in  P$ as described in \autoref{def:loop_pattern}}{
             
            $\powl_{do} \leftarrow$ \FMain{$\loopproject(N, \PTR(p_{do}, p_{redo}))$}\;
            $\powl_{redo} \leftarrow$ \FMain{$\loopproject(N, \PTR(p_{redo}, p_{do}))$}\;      
            \Return $\Loop(\powl_{do}, \powl_{redo})$\;
        }

        $G \leftarrow \popart(N)$ 
    
       \If{$(N, G)$ is a partial order pattern}{
            $\po \leftarrow \closure{\mathit{order}(N, G)}$\;
            \For{$T_i \in G = \{T_1, \dots, T_n\}$}{
                $\powl_i \leftarrow$ \FMain{$\poproject(N, T_i)$}\;
            }
            \Return $\po(\powl_1, \dots, \powl_n)$\;
        }

    \Return $null$\;
}
\end{algorithm}

\begin{figure}[!t]
    \centering    
     \includegraphics[width=0.35\textwidth]{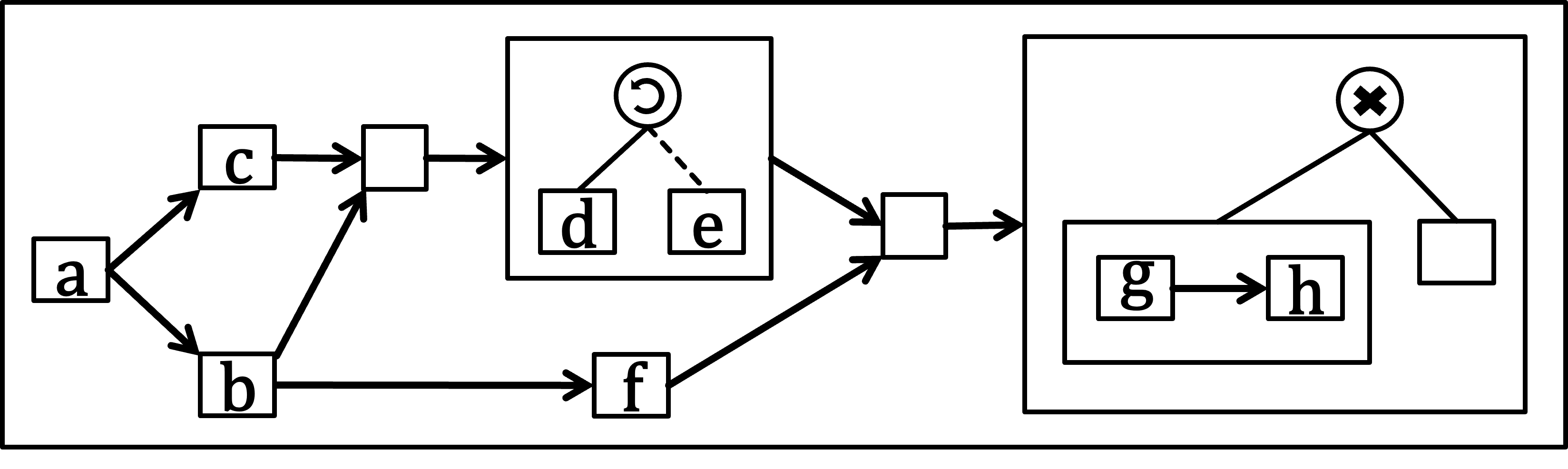}
    \caption{The POWL model generated by \autoref{alg:convert_pn_to_powl} for the WF-net from \autoref{fig:ex:wf}.\label{fig:ex:res}}
\end{figure}

\paragraph{Example Applications:} \autoref{fig:ex:res} shows the POWL model generated by applying \autoref{alg:convert_pn_to_powl} on the WF-net from \autoref{fig:ex:wf}. The generated POWL model, while structurally different from the manually crafted POWL model in \autoref{fig:ex:powl}, is semantically equivalent. \autoref{fig:neg_ex} illustrates three examples where \autoref{alg:convert_pn_to_powl} returns null:
\begin{itemize}
    \item The first WF-net (\autoref{fig:neg_ex:2}), while free-choice, its decision points are not arranged in a block structure, and no equivalent POWL model exists for its language. The algorithm attempts to find XOR or partial order patterns but ultimately produces partitions of size 1, leading to the return of null.
    \item The second WF-net (\autoref{fig:neg_ex:2}) exhibits a choice between activities $a$ and $b$, followed by a non-free-choice between $d$ and $e$ that is influenced by the preceding choice. This long-term dependency choice cannot be represented in POWL. \autoref{alg:convert_pn_to_powl} attempts to identify a partial order pattern, resulting in a partition that violates the requirements of \autoref{def:po_pattern}, e.g., no unique local end in the first part $T_1 = \{a, b\}$ since $\pre{p_3} \cap T_1 = \{a\} \neq \{b\} = \pre{p_4} \cap T_1$.
    \item In the third WF-net (\autoref{fig:neg_ex:3}), the places $p_2$ and $p_3$ model a choice between $d$ or executing both $b$ and $c$ concurrently. When attempting to identify a partial order pattern, the requirements are violated, e.g., no unique local start in the second part $T_2 = \{b,c,d\}$ since $\post{p_2} \cap T_2 = \{b, d\} \neq \{c, d\} = \post{p_3} \cap T_2$. However, applying the reduction rules from \autoref{fig:prepro} before applying \autoref{alg:convert_pn_to_powl} enables the successful conversion into POWL, as illustrated in \autoref{fig:ex_red}.
\end{itemize}  

\begin{figure}[!t]
\centering
    \begin{subfigure}{\textwidth}
    \centering
    \includegraphics[width=0.55\textwidth]{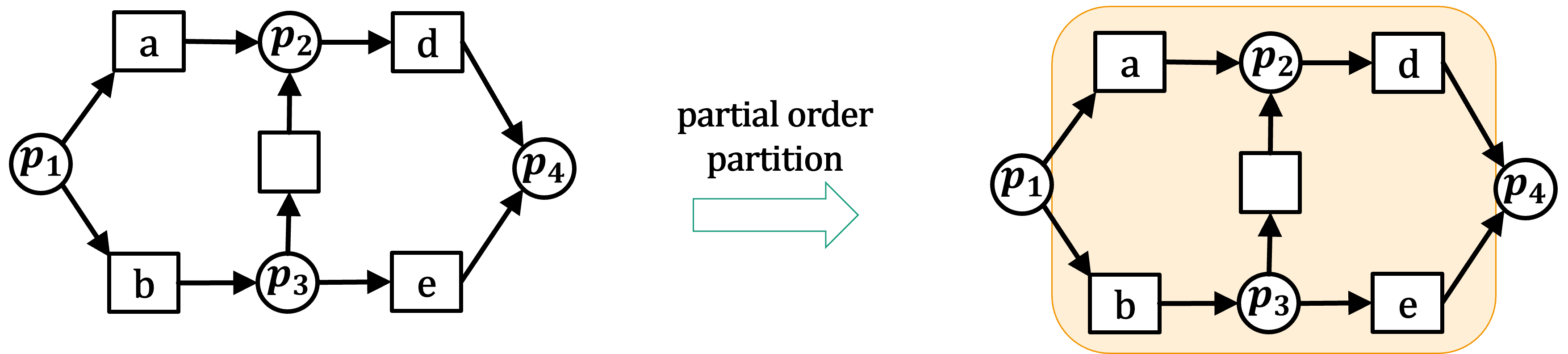}
    \caption{Free-choice WF-net with non-block-structured decision points.}\label{fig:neg_ex:1}
    \end{subfigure}
   
    \begin{subfigure}{\textwidth}
    \centering
    \includegraphics[width=0.68\textwidth]{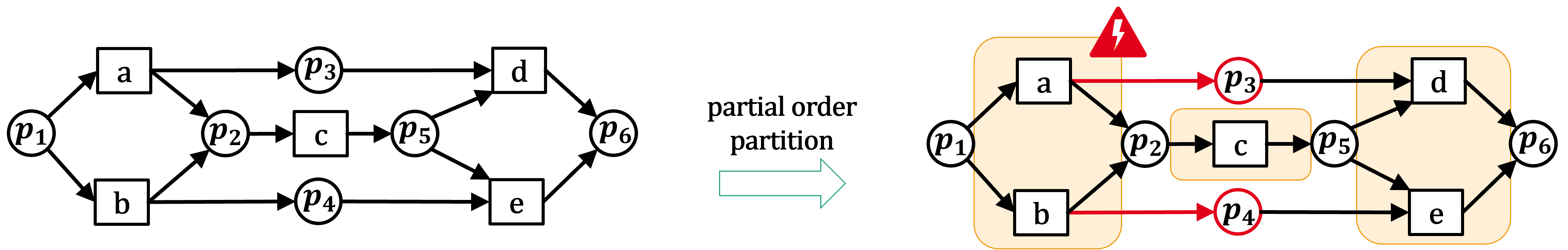}
    \caption{Non-free-choice WF-net with a long-term dependency between choices.}\label{fig:neg_ex:2}
    \end{subfigure}

    \begin{subfigure}{\textwidth}
    \centering
    \includegraphics[width=0.75\textwidth]{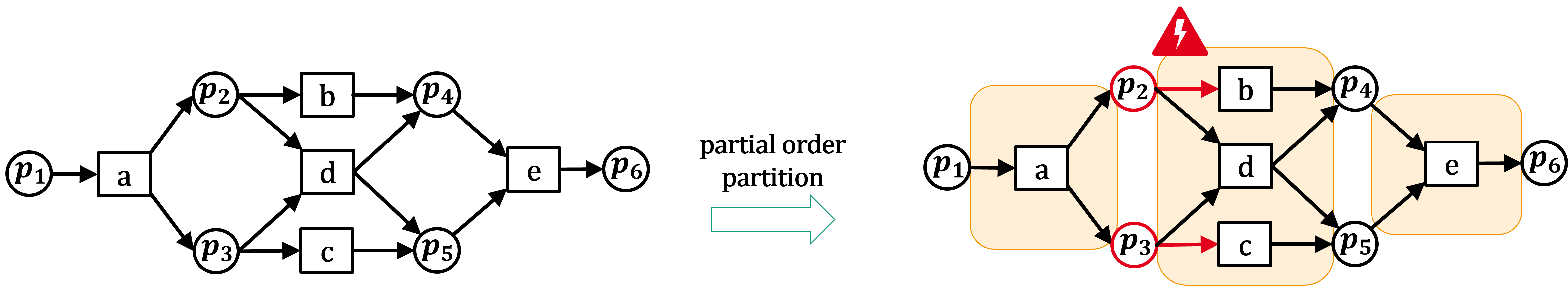}
    \caption{Non-free-choice WF-net where decision points combine choice with concurrency.}\label{fig:neg_ex:3}
    \end{subfigure}

    \caption{Examples where \autoref{alg:convert_pn_to_powl} returns null.\label{fig:neg_ex}}
\end{figure}

\begin{figure}[!t]
    \centering    
     \includegraphics[width=0.95\textwidth]{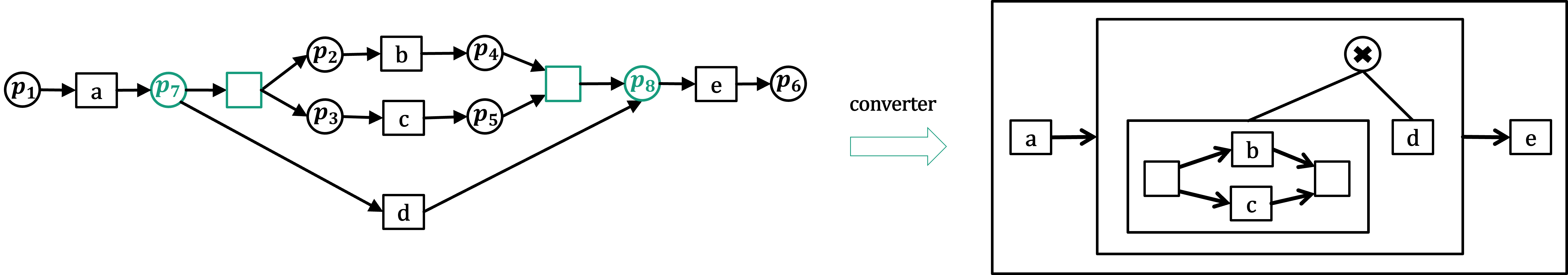}
    \caption{Result of applying the reduction rules (c.f. \autoref{fig:prepro}) to the WF-net from \autoref{fig:neg_ex:3}, followed by a successful conversion into POWL using \autoref{alg:convert_pn_to_powl}.\label{fig:ex_red}}
\end{figure}

\section{Correctness and Completeness Guarantees}\label{sec:gua}


In this section, we prove the correctness and completeness guarantees of \autoref{alg:convert_pn_to_powl}. Our proof strategy for correctness relies on structural induction on the WF-net. We show that for each pattern (XOR, loop, and partial order), the projection operation preserves safeness and soundness, allowing for the recursive application of the algorithm on the created subnets (\autoref{sec:lemmas:proj}). Then, we demonstrate that combining the languages of these subnets using the respective POWL operators accurately reflects the language of the original WF-net (\autoref{sec:lemmas:pattern}). We combine these findings to prove the overall correctness of the algorithm (\autoref{sec:gua:lang}). Finally, we show the completeness of our algorithm on semi-block-structured WF-nets (\autoref{sec:gua:redisc}).

\subsection{Projection Structural Guarantees}\label{sec:lemmas:proj}
This section proves that the XOR, loop, and partial order projections, when applied to safe and sound WF-nets, produce safe and sound WF-nets. This ensures that the recursive calls in \autoref{alg:convert_pn_to_powl} are always applied to valid inputs.

\begin{lemma}[XOR Projection Structural Guarantees]\label{thm:xor_soundness_preservation}
Let $(N, G)$ be an XOR pattern with $N = (P, T, F)$ and $G = \{T_1, \ldots, T_n\}$. Let $N_i = (P_i, T_i, F_i) \allowbreak = \xorproject(N, T_i)$ for each i ($1 \leq i \leq n$). Then $N_i$ is a safe and sound WF-net.
\end{lemma}

\begin{proof}

\textbf{(1) $N_i$ is a WF-net:} By construction, no place $p \in P_i \setminus \{N_{source}, N_{sink}\}$ is connecting transitions from both $T_i$ and $T \setminus T_i$ in $N$. Therefore, all transitions and places in $N_i$ remain connected on a path from $N_{source}$ to $N_{sink}$.





\textbf{(2) We prove that $N_i$ is safe and sound.}
$N_i$ can be replaced by a silent transition $t^{\texttt{+}} \notin T$ in $N$ generating another WF-net $N' = (P', T', F')$ as follows:
\begin{itemize}
    \item $T' = (T \setminus T_i) \cup \{t^{\texttt{+}}\}$.
    \item $P' = P \project{T \setminus T_i}$.
    \item $F' = F\project{P', T \setminus T_i} \cup \{(N_{source}, t^{\texttt{+}}),  (t^{\texttt{+}}, N_{sink})\}$.
    
\end{itemize}

$N$ can be reconstructed from $N'$ as described in \cite[Theorem 3.4]{compTheorem} by replacing $t^{\texttt{+}}$ by $N_i$ if $P_i \cap P' = \{N_{source}, N_{sink}\}$. This condition is satisfied since no place connecting transitions from both $T_i$ and $T \setminus T_i$ exists in $N$. Therefore, $N_i$ is safe and sound by \cite[Theorem 3.4]{compTheorem}. \end{proof}

\begin{lemma}[Loop Projection Structural Guarantees]\label{thm:loop_soundness_preservation}
Let $(N, G)$ be an loop pattern with $N = (P, T, F)$; $T_{do}, T_{redo} \in G$; and $p_{do}, p_{redo} \in P$ as defined in \autoref{def:loop_pattern}. Let $N_{do} = (P_{do}, T_{do}, F_{do}) = \loopproject(N, T_{do})$ and $N_{redo} = (P_{redo}, T_{redo}, F_{redo}) = \loopproject(N, T_{redo})$. Then $N_{do}$ and $N_{redo}$ are safe and sound WF-nets.
\end{lemma}

\begin{proof}

\textbf{(1) We show that $N_{do}$ is a WF-net (analogous proof for $N_{redo}$).}

\textbf{(1.1) Unique source and sink:} Assume for contradiction that there exists a place $p \in P_{do} \setminus \{N_{source}, N_{sink}\}$ such that $\pre{p} = \emptyset$ or $\post{p} = \emptyset$ in $N_{do}$. This place must be connecting transitions from both $T_{do}$ and $T_{redo}$ since $\pre{p} \neq \emptyset$ and $\post{p} \neq \emptyset$ in $N$. Without loss of generality, assume that there exist transitions $t \in T_{do}$ and $t' \in T_{redo}$ where $(t,p) \in F$ and $(p,t') \in F$. Since $t \in T_{do} = \PTR(p_{do}, p_{redo})$, there exists a path starting from $p_{do}$ to $t$, without passing through $p_{redo}$ in-between. From $t$, the path can continue through $p$ to reach $t'$. From $t'$, we know that we can eventually reach $p_{redo}$ (we reach $p_{do}$ first and then take any path from $p_{do}$ to $p_{redo}$). Combining all of these sequences ($p_{do} \rightarrow \ldots \rightarrow t \rightarrow p \rightarrow t' \rightarrow ... \rightarrow p_{redo})$ implies that $t' \in T_{do}$. This contradicts our assumption that $t' \in T_{redo}$.

\textbf{(1.2) Connectivity:} We showed that no place is connecting transitions from both $T_{do}$ and $T_{redo}$ except the loop entry and exit places. Therefore, all transitions and places in $N_{do}$ remain connected on a path from $N_{source}$ to $N_{sink}$.

\textbf{(2) Safeness and soundness:} The proof that $N_{do}$ and $N_{redo}$ are safe and sound is analogous to the safeness and soundness proof of \autoref{thm:xor_soundness_preservation}. \end{proof}

\begin{lemma}[Partial Order Projection Structural Guarantees]\label{thm:po_soundness_preservation}
Let $(N, G)$ be a partial order pattern with $N = (P, T, F)$ and $G = \{T_1, \ldots, T_n\}$. Let $N_i = (P_i, T_i, F_i) \allowbreak = \poproject(N, T_i)$ for each i ($1 \leq i \leq n$). Then $N_i$ is a safe and sound WF-net.
\end{lemma}

\begin{proof}

\textbf{(1) $N_i$ is a WF-net:} By construction, every node in $N_i$ is on a path from $p_s$ and $p_e$. The applied normalization ensures that additional source and/or sink places are inserted in case $\pre{p_s} \neq \emptyset$ and/or $\post{p_e} \neq \emptyset$, respectively.




\textbf{(2) We prove that $N_i$ is sound.}
 
\textbf{(2.1) No dead transitions:}
    Consider any transition $t \in T_i$. Since $N$ is sound, there exists a reachable marking $M$ in $N$ that enables $t$. Consider the marking $M_i$ in $N_i$ defined as follows for $p \in P_i$:
\begin{equation*}
M_i(p) =
    \begin{cases}
        M(p) & \text{if } p \in P, \\
        1 & \text{if } p = p_s \text{ and } M(p) = 1 \text{ for all } p \in \Pre{T_i}, \\
        1 & \text{if } p = p_e \text{ and } M(p) = 1 \text{ for all } p \in \Post{T_i}, \\
        0 & \text{otherwise (for places added by normalization)} .
    \end{cases}
\end{equation*}
        
    The projection operation only removes places and transitions that are not in $T_i$ and replaces the connections to $\Pre{T_i}$ and $\Post{T_i}$ with $p_s$ and $p_e$, respectively. Thus, any firing sequence leading to $M$ in $N$ can be transformed into a firing sequence leading to $M_i$ in $N_i$ by removing transitions not in $T_i$ (and potentially  firing the additional silent transition added by normalization). Therefore, $M_i$ is reachable from $[{N_i}_{source}]$ in $N_i$. The unique local start and end properties (c.f. \autoref{def:po_pattern}) ensure that if $t$ needs to consume a token from a place $p \in \Pre{T_i}$ (or $p \in \Post{T_i}$) in $N$, then all other places in $\Pre{T_i}$ (or in $\Post{T_i}$, respectively) must be included in $M$ as well. This implies that $M$ and $M_i$ agree on all places that are needed to enable $t$, including start and end places. Therefore, $t$ is enabled in $M_i$.


\textbf{(2.2) Option to complete:}
    Consider any marking $M_i$ reachable from $[{N_i}_{source}]$ in $N_i$. Due to the unique local start and end properties (c.f. \autoref{def:po_pattern}), there must exist a reachable marking $M$ in $N$ that enables the transitions in $T_i$ in the same way as $M_i$ does in $N_i$, i.e., $M$ is defined for $p \in P\project{T_i}$ as follows:
\begin{equation*}
    M(p) = 
    \begin{cases}
        M_i(p_s) & \text{if } p \in \Pre{T_i}, \\
        M_i(p_e) & \text{if } p \in \Post{T_i}, \\
        M_i(p) & \text{otherwise.}
    \end{cases}
\end{equation*} 

Since $N$ is sound, there exists a firing sequence $\sigma$ from $M$ to $[N_{sink}]$ in $N$. We can construct a corresponding firing sequence $\sigma_i$ from $M_i$ to $[{N_i}_{sink}]$ in $N_i$ by taking only the transitions in $\sigma$ that belong to $T_i$ (and potentially  firing the additional silent transition added by normalization). 


\textbf{(3) $N_i$ is safe:}
Assume, for the sake of contradiction, that there exist a reachable marking $M_i$ in $N_i$ and a place $p \in P_i$ such that $M_i(p) \geq 2$. 
There must exist a reachable marking $M$ in $N$ that enables the transitions in $T_i$ in the same way as $M_i$ does in $N_i$ (c.f. the proof of ``option to complete''). Then there exists $p' \in P\project{T_i}$ such that $M(p') = M_i(p) \geq 2$. This violates the safeness of $N$. \end{proof}

\subsection{Pattern LanguagePreservation Guarantees}\label{sec:lemmas:pattern}
This section establishes the language preservation guarantees of the identified patterns. We prove that the XOR, loop, and partial order patterns, when translated into their corresponding POWL representations, result in POWL models that have the same language as the original WF-net.

\begin{lemma}[XOR Pattern Language Preservation]\label{thm:xor_lang_preservation}
Let $(N, G)$ be an XOR pattern and $G = \{T_1, \ldots, T_n\}$. Let $\powl_1, \ldots, \powl_n$ be POWL models  such that $\lang(\powl_i) = \lang(\xorproject(N, T_i))$ for each $i$ ($1 \leq i \leq n$). Then, $\lang(\xor(\powl_1, \ldots, \powl_n)) = \lang(N)$.
\end{lemma}

\begin{proof}
Since $(N, G)$ forms an XOR Pattern, $N$ enforces a choice of transitions from exactly one part $T_i \in G$ to be executed. $\xorproject(N, T_i)$ captures exactly the executions involving the transitions in $T_i$. Therefore, we conclude:
\[
\lang(N) = \bigcup_{i=1}^n \lang(\xorproject(N, T_i)).
\]


By combining this equality with \autoref{def:lang} and the assumption that $\lang(\powl_i) = \lang(\xorproject(N, T_i))$ for each $1 \leq i \leq n$, we get: $\lang(N) = \lang(\xor(\powl_1, \ldots, \powl_n))$.


\end{proof}

\begin{lemma}[Loop Pattern Language Preservation]\label{thm:loop_lang_preservation}
Let $(N, G)$ be a loop pattern with $N = (P, T, F)$; $T_{do}, T_{redo} \in G$; and $p_{do}, p_{redo} \in P$ as defined in \autoref{def:loop_pattern}. Let $\powl_{do}$ and $\powl_{redo}$ be POWL models such that $\lang(\loopproject(N, T_{do})) = \lang(\powl_{do})$ and $\lang(\loopproject(N, T_{redo})) = \lang(\powl_{redo})$. Then, $\lang(\Loop(\powl_{do}, \powl_{redo})) = \lang(N)$.
\end{lemma}

\begin{proof}
Let $N_{do} = \loopproject(N, T_{do})$ and $N_{redo} = \loopproject(N, T_{redo})$. Due to the soundness of $N_{do}$ and $N_{redo}$ and the absence of places connecting transitions from both the do- and redo-parts except $p_{do}$ and $p_{redo}$ (c.f., the proof of \autoref{thm:loop_soundness_preservation}), we conclude that the language of $N$ consists of sequences that can be segmented into subsequences of complete executions of $N_{do}$ and $N_{redo}$, interleaved as follows:
\[
\lang(N) = \lang(N_{do}) \cdot (\lang(N_{redo}) \cdot \lang(N_{do}))^*.
\]

By combining this equality with \autoref{def:lang} and the assumption that $\lang(\powl_{do}) = \lang(N_{do}) \ \wedge \ \lang(\powl_{redo}) = \lang(N_{redo})$, we get: $\lang(N) = \lang(\Loop(\powl_{do},\powl_{redo}))$.

\end{proof}

\begin{lemma}[Partial Order Pattern Language Preservation]\label{thm:po_lang_preservation}
Let $(N, G)$ be a partial order pattern with $N = (P, T, F)$, $G = \{T_1, \ldots, T_n\}$, and $\po \in \Orders{n}$ as defined in \autoref{def:po_pattern}. Let $\powl_1, \ldots, \powl_n$ be POWL models such that $\lang(\poproject(N, T_i)) \allowbreak = \lang(\powl_i)$ for each $i$ ($1 \leq i \leq n$). Then, $\lang(\po(\powl_1, \ldots, \powl_n)) = \lang(N)$.
\end{lemma}

\begin{proof}
Let $N_i = \poproject(N, T_i)$ for each $i$ ($1 \leq i \leq n$). By combining  the semantics of partial orders (c.f. \autoref{def:lang}) with the assumption that $\lang(\powl_i) = \lang(N_i)$ for each $i$ ($1 \leq i \leq n$), we can write:
\[
\lang(\po(\powl_1, \ldots, \powl_n)) = \{\sigma \in \shuffle_{\po}(\sigma_1 , ..., \sigma_n) \ | \ \forall_{1 \leq i \leq n} \sigma_i \in \lang(N_i)\}.
\]


\textbf{(1) Proof for $\lang(N) \subseteq \{\sigma \in \shuffle_{\po}(\sigma_1 , ..., \sigma_n) \ | \ \forall_{1 \leq i \leq n} \sigma_i \in \lang(N_i)\}$:} Let $\sigma \in \lang(N)$ be any firing sequence of $N$. We construct subsequences $\sigma_1, ..., \sigma_n$ by projecting $\sigma$ onto $T_i$ for each $i$ ($1 \leq i \leq n$). We need to show that $\sigma$ can be expressed as a shuffle of $\sigma_1, ..., \sigma_n$, respecting the partial order $\po$. This can be derived by proving the following three key points:

\begin{itemize}
    \item All parts $T_i \in G$ are present within $\sigma$.
    \item Each $\sigma_i$ is a firing sequence from $N_i$ (i.e., $\sigma_i \in \lang(N_i)$). 
    \item The partial order between the subsequences is preserved in $\sigma$  (i.e., $\sigma \in \shuffle_{\po}(\sigma_1 , ..., \sigma_n)$).
\end{itemize}

\textbf{(1.1) All parts $T_i \in G$ are present within $\sigma$:} Assume, for the sake of contradiction, that there exists a part $T_i \in G$ not present within $\sigma$, i.e., $\sigma_i = \langle \rangle$. Since $N$ is a WF-net, this implies that there must be a \emph{decision point} $p \in P$ where some outgoing paths from $p$ eventually lead to transitions in $T_i$ and other outgoing paths from $p$ do not. This violates the first condition of \autoref{def:po_pattern}, which states that if a place has outgoing flows leading to different transitions, then all such transitions must fall into the same part in $G$. 

\textbf{(1.2) 
Each $\sigma_i$ is a firing sequence from $N_i$:}
The unique local start and end properties (c.f. \autoref{def:po_pattern}) ensure each subnet is executed independently from start to end within $N$. Therefore, $\sigma_i$ must be a firing sequence from $N_i$.


\textbf{(1.3) The partial order $\po$ is preserved in $\sigma$:} 
Assume $i \po j$ for any $i, j \in \{1, \dots, n\}$. Since $\po = \closure{\mathit{order}(N, G)}$, transitions from $T_i$ are executed first, producing tokens that are needed to eventually enable $T_j$. Suppose, for the sake of contradiction, that after the execution of transitions from $T_j$, transitions from $T_i$ are re-enabled (i.e., tokens are produced in $\Pre{T_i}$). Then we have two possible scenarios:
    \begin{itemize}
        \item (i) The re-enabling of $T_i$ does not depend on the completion of $T_j$ (i.e., it does not require the consumption of tokens from $\Post{T_j}$): This means that we can perform a full execution of the subnet of $T_i$ and reach the subnet of $T_j$ again before its completion, violating safeness.
        \item (ii) The re-enabling of $T_i$ depends on the completion of $T_j$ (i.e., it requires the consumption of tokens from $\Post{T_j}$): This implies the existence of a sequence of dependencies in the execution order $j \po \dots \po i$. By transitivity, $j \po i$ holds. This violates the asymmetry requirement of partial orders since $i \po j$.
    \end{itemize}    


\textbf{(2) Proof for $\{\sigma \in \shuffle_{\po}(\sigma_1 , ..., \sigma_n) \ | \ \forall_{1 \leq i \leq n} \sigma_i \in \lang(N_i)\} \subseteq \lang(N)$:}
Consider any sequence $\sigma \in \shuffle_{\po}(\sigma_1 , ..., \sigma_n)$ where $\sigma_i \in \lang(N_i)$ for $1 \leq i \leq n$. We showed that all parts $T_i \in G$ must be visited in $N$ (c.f. the proof of 1.1). Due to the unique local start and end properties in $N$ (c.f. \autoref{def:po_pattern}), each subnet can be executed independently in $N$, without violating the execution order $\po = \closure{\mathit{order}(N, G)}$. Therefore, the interleaved sequence $\sigma$ constitutes a valid firing sequence in $N$. 


\end{proof}

\subsection{Overall Correctness Guarantee}\label{sec:gua:lang}
In this section, we prove the correctness of \autoref{alg:convert_pn_to_powl}. Specifically, we show that the algorithm, if successfully producing a POWL model, then the POWL model has the same language as the input WF-net.

\begin{theorem}[Correctness]\label{thm:main_correctness_compact}
Let $N = (P, T, F)$ be a safe and sound WF-net. If \autoref{alg:convert_pn_to_powl} successfully converts $N$ into a POWL model $\powl$, then $\lang(N) = \lang(\powl)$.
\end{theorem}

\begin{proof}
We prove the theorem by induction on the number of transitions in $N$.

\begin{itemize}

\item \textbf{Base case:} The theorem trivially holds for a WF-net that contains a single transition.

\item \textbf{Inductive hypothesis:} For $n> 1$, assume the theorem holds for all safe and sound WF-net with fewer transitions than $n$ (i.e., $|T| < n$).

\item \textbf{Inductive step ($|T| = n$):} We consider the different cases in the algorithm:

\begin{itemize}

\item \textbf{XOR pattern:} Suppose an XOR pattern $(N, G)$ with $G = \{T_1, \dots, T_n\}$ is identified. For each $T_i \in G$, $N$ is projected onto $T_i$ to obtain $N_i = \xorproject(N, T_i)$.
 By \autoref{thm:xor_soundness_preservation}, each $N_i$ is a safe and sound WF-net with fewer transitions than $n$.
 By the inductive hypothesis, the POWL model $\powl_i$ obtained from $N_i$ satisfies $\lang(\powl_i) = \lang(N_i)$.
 The algorithm returns $\powl = \xor(\powl_1, \dots, \powl_n)$.
 By \autoref{thm:xor_lang_preservation}, $\lang(\powl) = \lang(N)$.

 \item \textbf{Loop or partial order pattern:} The proof is analogous to the XOR case, using the appropriate lemmas for structural guarantees (\autoref{thm:loop_soundness_preservation} or \autoref{thm:po_soundness_preservation}) and language equivalence (\autoref{thm:loop_lang_preservation} or \autoref{thm:po_lang_preservation}).



    \item \textbf{No pattern is detected:} If no pattern is detected, then the algorithm returns $null$, which does not contradict the theorem.

\end{itemize}

\end{itemize}

By induction, the theorem holds for all safe and sound WF-nets successfully converted by \autoref{alg:convert_pn_to_powl}.

\end{proof}

\subsection{Completeness Guarantee on Semi-Block-Structured WF-Nets} \label{sec:gua:redisc}
In this section, we show that \autoref{alg:convert_pn_to_powl} is complete when applied to semi-block-structured WF-nets (c.f. \autoref{def:block}).

\begin{theorem}[Completeness]\label{thm:complete}
Let $N$ be a semi-block-structured WF-net. \autoref{alg:convert_pn_to_powl} successfully converts $N$ into a POWL model that has the same language as $N$. 
\end{theorem}

\begin{proof}
We distinguish between three cases:
\begin{itemize}
    \item \textbf{Case 1: $N$ has a single transition:} This matches the base case of the algorithm.
    \item \textbf{Case 2: $N$ corresponds to an XOR or loop pattern:} Projecting the net on each part yields another semi-block-structured WF-net.
    \item \textbf{Case 3: $N$ does not correspond to an XOR or loop pattern:} The algorithm creates a partition $G = \popart(N)$ where transitions within the same block are grouped into the same part of the partition. The transitive closure of the execution order $\closure{\mathit{order}(N, G)}$ must form a partial order because (i) substituting each part with a single transition turn the net into a marked graph and (ii) soundness implies acyclicity for marked graph WF-nets. Since the top-level blocks have unique entry and exit points, the unique local start and end requirements of \autoref{def:po_pattern} are also met. Therefore, a partial order pattern is detected. Projecting $N$ on each part yields either (i) a base case for single transitions or (ii) a semi-block-structured WF-net that corresponds to an XOR or loop pattern for the blocks.
\end{itemize}
After projection, sub-nets are recursively handled in the same manner. Thus, the algorithm successfully produces a POWL model. The language equivalence follows by \autoref{thm:main_correctness_compact}. 

\end{proof}

\section{Implementation and Scalability Assessment}\label{sec:eval}
To assess the scalability of \autoref{alg:convert_pn_to_powl}, we implemented it, incorporating the reduction rules illustrated in \autoref{fig:prepro} and \autoref{fig:loop_prepro}. We then performed two experiments (code and data are available at \url{https://github.com/humam-kourani/WF-net-to-POWL}). In the first experiment, we utilized the process tree generator from \cite{DBLP:conf/bpm/JouckD16,DBLP:journals/bise/JouckD19} to generate $1000$ process trees. These trees were then translated into WF-nets using PM4Py \cite{DBLP:journals/simpa/BertiZS23}, resulting in a diverse set of WF-nets varying in size from $21$ to $370$ transitions and $15$ to $305$ places. In the second experiment, we used the ground truth WF-nets of the $20$ processes from \cite{DBLP:journals/corr/abs-2412-00023}, which were originally derived from POWL models. For comparison, we also applied the WF-net to process tree converter from \cite{DBLP:journals/algorithms/ZelstL20} in both experiments.

\paragraph*{Results.} In the first experiment, \autoref{alg:convert_pn_to_powl} successfully generated POWL models for all $1000$ WF-nets, which was expected since process trees represent a subclass of POWL. The experiment demonstrated the high scalability of our approach, as illustrated in \autoref{fig:ev}. Conversion times for our approach ranged from $0.002$ to $2.48$ seconds, whereas the tree-based converter from \cite{DBLP:journals/algorithms/ZelstL20} required between $0.013$ and $126$ seconds. In the second experiment, while our algorithm was successful on all WF-nets, the tree-based converter was only able to convert $17$ out of the $20$ WF-nets into process trees. In summary, our experiments highlight the advantage of our approach in both supporting a broader range of structures and its superior performance compared to the process tree-based converter.

\begin{figure}[!t]
    \centering    
    \includegraphics[width=0.64\textwidth]{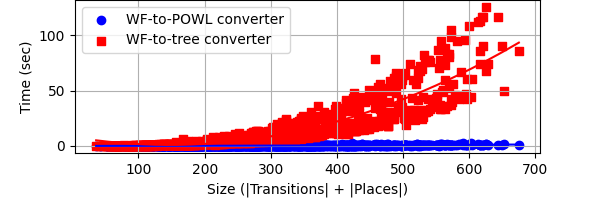}
    \caption{Comparison of conversion times between \autoref{alg:convert_pn_to_powl} and the process tree converter from \cite{DBLP:journals/algorithms/ZelstL20} on the 1000 WF-nets dataset.}\label{fig:ev}
\end{figure}

Note that the implemented algorithm is also available in ProMoAI \cite{DBLP:conf/ijcai/KouraniB0A24} (\url{https://promoai.streamlit.app/}), powering the redesign feature for improving existing process models via large language models.

\section{Conclusion}\label{sec:conc}
This paper introduced a novel algorithm for translating safe and sound Workflow Nets (WF-nets) into the Partially Ordered Workflow Language (POWL). The algorithm leverages the hierarchical structure of POWL by recursively identifying patterns within the WF-net that correspond to POWL's operators. We formally proved the correctness of our approach, showing that the resulting POWL model preserves the language of the original WF-net. Furthermore, we demonstrated the high scalability of the proposed algorithm and showed its completeness on semi-block-structured WF-nets, a subclass that contains equivalent workflow-nets for any POWL model. 

This work paves the way for broader adoption of POWL in different process mining applications. The main avenue for future work is the development of optimized process mining techniques that leverage the structural properties of POWL, such as efficient algorithms for conformance checking. Furthermore, we aim to provide a platform that allows users to visualize process models in the POWL language and implements new methods for interactive process analysis and improvement.

\bibliographystyle{plain}
\bibliography{lit}

\end{document}